\documentclass{article}

\usepackage{arxiv}
\usepackage[utf8]{inputenc} 
\usepackage[T1]{fontenc}    
\usepackage{hyperref}       
\usepackage{url}            
\usepackage{booktabs}       
\usepackage{amsfonts}       
\usepackage{nicefrac}       
\usepackage{microtype}      
\usepackage{lipsum}
\usepackage{graphicx}
\graphicspath{ {./images/} }

\usepackage{comment}
\usepackage{amsmath}
\usepackage{algorithm}
\usepackage{algpseudocode}
\usepackage{bm}
\usepackage{cite}
\usepackage{amsthm}
\usepackage{thm-restate}
\usepackage[english]{babel}

\theoremstyle{plain}
\newtheorem{theorem}{Theorem}
\ifx\numberwithinsect\relax
  \numberwithin{theorem}{section}
\fi
\ifx\useautoref\relax
	\addto\extrasenglish{%
	}
	\addto\extrasUKenglish{%
	}
	\addto\extrasUSenglish{%
	}
	\ifx\usethmrestate\relax
		\newtheorem{lemma}[theorem]{Lemma}
		\newtheorem{corollary}[theorem]{Corollary}
		\newtheorem{proposition}[theorem]{Proposition}
		\newtheorem{exercise}[theorem]{Exercise}
		\newtheorem{definition}[theorem]{Definition}
		\newtheorem{conjecture}[theorem]{Conjecture}
		\newtheorem{observation}[theorem]{Observation}
		\theoremstyle{definition}
		\newtheorem{example}[theorem]{Example}
		\theoremstyle{remark}
		\newtheorem{note}[theorem]{Note}
		\newtheorem*{note*}{Note}
		\newtheorem{remark}[theorem]{Remark}
		\newtheorem*{remark*}{Remark}
		\theoremstyle{claimstyle}
		\newtheorem{claim}[theorem]{Claim}
		\newtheorem*{claim*}{Claim}
	\else
		\newaliascnt{lemma}{theorem}
		\newtheorem{lemma}[lemma]{Lemma}
		\aliascntresetthe{lemma}
		
		\newaliascnt{corollary}{theorem}
		\newtheorem{corollary}[corollary]{Corollary}
		\aliascntresetthe{corollary}
		
		\newaliascnt{proposition}{theorem}
		\newtheorem{proposition}[proposition]{Proposition}
		\aliascntresetthe{proposition}
		
		\newaliascnt{exercise}{theorem}
		\newtheorem{exercise}[exercise]{Exercise}
		\aliascntresetthe{exercise}
		
		\newaliascnt{definition}{theorem}
		\newtheorem{definition}[definition]{Definition}
		\aliascntresetthe{definition}
		
		\newaliascnt{conjecture}{theorem}
		\newtheorem{conjecture}[conjecture]{Conjecture}
		\aliascntresetthe{conjecture}
		
		\newaliascnt{observation}{theorem}
		\newtheorem{observation}[observation]{Observation}
		\aliascntresetthe{observation}
		
		\theoremstyle{definition}
		\newaliascnt{example}{theorem}
		\newtheorem{example}[example]{Example}
		\aliascntresetthe{example}
		
		\theoremstyle{remark}
		\newaliascnt{note}{theorem}
		\newtheorem{note}[note]{Note}
		\aliascntresetthe{note}
		
		\newtheorem*{note*}{Note}
		\newaliascnt{remark}{theorem}
		\newtheorem{remark}[remark]{Remark}
		\aliascntresetthe{remark}
		
		\newtheorem*{remark*}{Remark}
		\theoremstyle{claimstyle}
		\newaliascnt{claim}{theorem}
		\newtheorem{claim}[claim]{Claim}
		\aliascntresetthe{claim}
		
		\newtheorem*{claim*}{Claim}
		\ifx\numberwithinsect\relax
			\numberwithin{theorem}{section}
			\numberwithin{lemma}{section}
			\numberwithin{corollary}{section}
			\numberwithin{proposition}{section}
			\numberwithin{exercise}{section}
			\numberwithin{definition}{section}
			\numberwithin{conjecture}{section}
			\numberwithin{observation}{section}
			\numberwithin{example}{section}
			\numberwithin{note}{section}
			\numberwithin{remark}{section}
			\numberwithin{claim}{section}
		\fi
	\fi
\else
	\newtheorem{lemma}[theorem]{Lemma}

	\theoremstyle{definition}
	
	\theoremstyle{remark}
	
	\newtheorem*{note*}{Note}
	
	\newtheorem*{remark*}{Remark}
	\theoremstyle{claimstyle}
	
	\newtheorem*{claim*}{Claim}
\fi
\theoremstyle{plain}

\addto\extrasenglish{
  
}

\title{Job Scheduling under Base and Additional Fees, with Applications to Mixed-Criticality Scheduling}

\author{
  Yi-Ting Hsieh \\
  Institute of Data Science and Engineering\\
  National Yang-Ming Chiao-Tung University, Taiwan \\
  \texttt{xuryp891009.cs12@nycu.edu.tw} \\
  \And
  Mong-Jen Kao \\
  Department of Computer Science\\
  National Yang-Ming Chiao-Tung University, Taiwan\\
  0000-0002-7238-3093 \\
  \texttt{mjkao@nycu.edu.tw} \\
  \And
  Jhong-Yun Liu \\
  Institute of Computer Science and Engineering\\
  National Yang-Ming Chiao-Tung University, Taiwan\\
  \texttt{johnliu.cs12@nycu.edu.tw} \\
  \And
  Hung-Lung Wang \\
  Department of Computer Science\\ and Information Engineering\\
  National Taiwan Normal University, Taipei, Taiwan\\
  0000-0001-6156-2734\\
  \texttt{hlwang@ntnu.edu.tw} \\
}

\begin{document}
\maketitle
\begin{abstract}
We are concerned with the problem of scheduling $n$ jobs onto $m$ identical machines. Each machine has to be in operation for a prescribed time, and the objective is to minimize the total machine working time. Precisely, let $c_i$ be the prescribed time for machine $i$, where $i\in[m]$, and $p_j$ be the processing time for job $j$, where $j\in[n]$. The problem asks for a schedule $\sigma\colon\, J\to M$ such that $\sum_{i=1}^m\max\{c_i, \sum_{j\in\sigma^{-1}(i)}p_j\}$ is minimized, where $J$ and $M$ denote the sets of jobs and machines, respectively. We show that First Fit Decreasing (FFD) leads to a $1.5$-approximation, and this problem admits a polynomial-time approximation scheme (PTAS). The idea is further applied to 
mixed-criticality system scheduling to yield improved approximation results. 
\end{abstract}

\section{Introduction} \label{sec:intro}

In many service models, there are two kinds of service charge: a base fee and an additional usage fee. If the customer requests the service within the prescribed time, only the base fee will be charged; otherwise, additional usage fees will apply based on how long the service is extended. 
For example, cloud computing services, cloud storage, car rental, and online courses are all services of this kind. 
In this service model, we are concerned with a scheduling problem that aims at minimizing the service charge, consisting of a base fee and an additional fee if the service is to be extended. 
The given are $n$ jobs and $m$ machines. Each job $j$ has its processing time $p_j$, and every machine needs time $p_j$ to process job $j$. For each machine $i$, there is a base operation time $c_i$. The goal is to assign each job to a machine so that the \textit{total machine working time} is minimized. A machine, say $i$, always lasts for time $c_i$, unless the jobs assigned to $i$, say $j_{i_1},\dots, j_{i_k}$, need more time to finish. In such a case, machine $i$ lasts for time $\sum_{l=1}^{k}p_{j_{i_l}}$. Precisely, let $J$ and $M$ be the sets of jobs and machines, respectively. For any assignment $\sigma\colon\, J\to M$, the objective function is
\[\sum_{i=1}^m \max\left\{c_i, \sum_{j\in \sigma^{-1}(i)}p_j\right\},\]
where $\sigma^{-1}(i) = \{j\in J\colon\, \sigma(j)=i\}$.
We shall refer this problem to as Base and Additional Fee (BAF) Job Scheduling problem. 
The problem is NP-hard, due to a straightforward reduction from the Subset Sum Problem~\cite[A3.2, SP13]{MR519066}. Moreover, BAF is NP-hard in the strong sense~\cite{HanzalekTunysSucha16}, as we will see later in this section.

\smallskip
The BAF job scheduling problem has a ``min-sum'' objective, and there are typical scheduling problems whose objectives are of this kind; e.g., to minimize the sum of tardiness or earliness~\cite{BakerScudder1990,JDJL1990,Emmons1969, KoSt2007}, to minimize the sum of (weighted) job completion times~\cite{LENSTRA1977343,Li2020}, and to minimize the number of late jobs~\cite{Moore1968}. However, none of the formulations reveal the nature of 
the investigated one. 

\paragraph*{Non-Preemptive Mixed-Criticality Scheduling}

One closely related problem of BAF job scheduling is scheduling in \textit{non-preemptive mixed-criticality uni-processor systems with match-up}~\cite{HanzalekTunysSucha16,NOVAK2019687}. 
In this model, the execution of a job comes with different criticality levels, which is interpreted as different states of emergency.
When the execution of a job enters a higher criticality level, it takes a longer period of time to complete and meanwhile prevents jobs with lower criticality levels from being executed.
Due to this setting, the jobs are referred to as \emph{F-shaped} in the literature~\cite{HanzalekTunysSucha16,NOVAK2019687}.
The general objective in this scenario is to compute a uni-processor scheduling of the jobs, such that the supposed execution of each job at all criticality levels are aligned and the worst case job completion time is minimized.
See Figure~\ref{fig-mix-cri-1} for an illustration of such a system.
Notice that the supposed execution of low-criticality jobs may be discarded when certain preceding job enters a higher criticality level.

\begin{figure*}[h]
\centering
\includegraphics[scale=0.86]{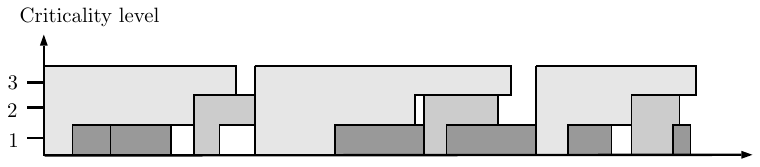}
\caption{An illustration of a schedule for a system with three criticality levels.}
\label{fig-mix-cri-1}
\end{figure*}

Let $L$ denote the maximum criticality level (height) of the jobs.
Hanz\'{a}lek et al.~\cite{HanzalekTunysSucha16} showed that the problem is strongly NP-hard, even for the case with $L=2$. 
Later on, Novak~et~al.~\cite{NOVAK2019687} proposed a $2$-approximation for $L=2$, which generalizes to an $L$-approximation for $L \ge 3$. 

\smallskip

We note that, the proposed BAF scheduling problem can be formulated as a mixed-criticality scheduling problem with two criticality levels proved in and vice versa, by treating a machine in $M$ as a high-criticality job and a job in $J$ as a low-criticality one.
Hence, the BAF job scheduling problem inherits the strongly NP-hardness of two-level mixed-criticality scheduling proved in~\cite{HanzalekTunysSucha16}.

\smallskip

Different job models in mixed-criticality systems have been considered in the literature.
See~\cite{BrunsDavis2017} for a comprehensive survey.
In particular, when the processing time of any job is equal to its criticality level, this model is referred as \emph{isosceles right-angle triangle scheduling}~\cite{BALOGH2024106718,durr2016}. 
For this problem, D\"{u}rr~et~al.~\cite{durr2016} proposed a quasi-polynomial-time approximation scheme (QPTAS), showing that this variant is not APX-complete. They also showed that the \textit{Best-Fit Decreasing} (BFD) strategy leads to an approximation algorithm with a ratio between $1.05$ and $1.5$. Recently, Balogh et al.~\cite{BALOGH2024106718} improved the lower bound on the ratio to be $1.27$. Similar results were proposed for arranging circles~\cite{ABCCKKS_2018}.

\paragraph*{Further Related Works}

In addition to the objective of minimizing worst-case job completion time, the analysis of schedulability has been conducted when each job comes with a release time and a deadline in real-time systems, see, e.g.,~\cite{Vestal07,BaruahLS10,BaruahBDLMMS12} and~\cite{BrunsDavis2017} for a comprehensive survey.

For the BAF job scheduling problem, regardless of the objective function, problems with similar settings have been studied previously. 
For example, the \textit{capacitated scheduling problem}, in which each machine has a capacity, indicating the number of jobs allowed to be executed on that machine~\cite{YANG2003449,LinKlaus2016}. 
The \textit{multiple subset sum problem}~
\cite{CapraraKellerer2000} and, more generally, the \textit{multiple knapsack problem}~\cite{CAPRARA2000111,ChekuriKhanna2005,Jansen2007,Jansen2012} 
also belong to the category of job scheduling problems with constraints on the working durations. 
Notice that for the problems mentioned above, the ``capacities'' associated on the machines or bins serve as constraints. In contrast, the base operation times in our problem are parts of the objective function. 
\paragraph*{Our Results and Discussion}
Our main result focuses on the approximability of the Base and Additional Fee (BAF) Job Scheduling problem.
We show that
\begin{enumerate}
    \item the \textit{First Fit Decreasing} (FFD) strategy leads to a $1.5$-approximation, and
    \item this problem admits a polynomial-time approximation scheme (PTAS).
\end{enumerate}

We note that, one classic technique to obtaining an approximation scheme for problems of this nature is to round down the processing time of the jobs, schedule the ``big jobs'' separately, and then utilizes the property of ``small jobs'' to fill them in~\cite{durr2016,DBLP:journals/combinatorica/VegaL81}.

\smallskip

While this approach does lead to a quasi-PTAS for our considered problem, it does not improve further to a PTAS.
One of the main issues is the lacking of a proper upper-bound on the capacity of the machines, which makes the complexity of the big jobs too high to handle in polynomial-time.
To overcome this issue, we divide the jobs into groups of geometrically growing sizes and create appropriate ``gaps'' between them, so that consecutive groups can be handled relatively independently in polynomial-time while ensuring a small error in the guarantee.
Since this problem is strongly NP-hard, a PTAS is the best approximation result one can expect in polynomial-time.

\smallskip

Our results for the BAF job scheduling problem translate to mixed-criticality scheduling scheduling with $L=2$ since the two problems are equivalent in formulations.
For mixed-criticality scheduling with $L \ge 3$, we apply the PTAS algorithm presented for BAF scheduling to obtain an asymptotic $L(1+\epsilon)/2$-approximation.
\paragraph*{Organization of this Paper}
In~\autoref{sec:prob_def} we provide formal problem definition and notations. 
In~\autoref{sec:greedy} and~\autoref{sec:approx-schemes} we present our $3/2$-approximation and PTAS for the BAF job scheduling problem.
We describe the application to mixed-criticality system scheduling in~\autoref{sec:multi-case}.
Due to page limit, we provide omitted proofs in Section~\ref{sec-appendix-proof} in the appendix for further reference.
\section{Problem Definition and Notations} \label{sec:prob_def}

In the base and addition fee (BAF) job scheduling problem, the given is a 4-tuple $(J, M, p, c)$, where $J$ is the set of jobs, $M$ is the set of machines, $p$ is a processing time function for $J$, and $c$ is a base operation function for $M$.
We assume that $p_j\in \mathbb{R}_{>0}$ and $c_i\in \mathbb{R}_{>0}$ for all $j\in J$ and $i\in M$.
Given an instance $\Pi=(J,M,p,c)$, a schedule for $\Pi$ is a mapping $\sigma\colon\, J \to M$. 
The working time of a machine $i \in M$ in the schedule 
$\sigma$ is defined to be $T_{\Pi}(\sigma,i) := \max\{ \sum_{j\in \sigma^{-1}(i)} p_j, \; c_i \}$, where $\sigma^{-1}(i)$ denotes the set of jobs that are assigned to $i$.
The goal of this problem is to compute a schedule $\sigma$ such that the overall machine working time, $T_{\Pi}(\sigma) := \sum_{i \in M} T_{\Pi}(\sigma,i)$, is minimized. 

\smallskip
For succinctness we write $T(\sigma, i)$ and $T(\sigma)$ for $T_{\Pi}(\sigma, i)$ and $T_{\Pi}(\sigma)$, respectively,   when there is no ambiguity from the text. 
Due to the nature of prescribed base operation times, we will refer them to as the \emph{capacities} of the machines in this work.

\begin{figure*}[h]
\centering
\includegraphics[scale=0.86]{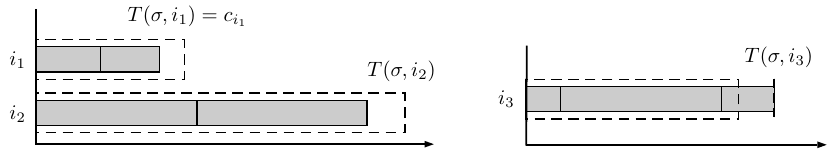}
\caption{An illustration of a BAF job scheduling with three machines. A dashed box represents a machine and its capacity, and a solid box represents a job and its processing time.}
\label{fig-two-layer}
\end{figure*}

\smallskip
In the non-preemptive mixed-criticality match-up scheduling problem, we are given an integer $L \in \mathbb N$ which denotes the maximum criticality level of the system and a set of F-shaped jobs ${\mathcal B}$, where each $B \in {\mathcal B}$ is associated with a height $h(B)$ with $1\le h(B) \le L$ and $h(B)$ layers of non-decreasing processing times 
$p(B,1) \le p(B,2) \le \dots \le p(B,h(B))$.

\smallskip
A schedule of the F-shaped jobs is a mapping $\sigma \colon {\mathcal B} \to {\mathbb R}_{\ge 0}$ where $\sigma(B)$ for each $B \in {\mathcal B}$ denotes the time-index at which the job is scheduled to execute.
The schedule $\sigma$ is feasible if for all layers $1\le i\le L$ and any pair of jobs $B, B' \in {\mathcal B}$ with heights at least $i$, 
i.e., $\min\{ h(B), h(B') \} \ge i$,
one of the following two conditions,
$$\sigma(B') + p(B',i) \; \le \; \sigma(B)  \quad \text{and} \quad \sigma(B) + p(B,i) \; \le \; \sigma(B'),$$
must hold.
Intuitively, we require that in a feasible schedule, the placement of any pair of F-shaped jobs must be disjoint in any layer.
See also Figure~\ref{fig-mix-cri-1} for an illustration.
Let $\sigma$ be a feasible schedule.
The maximum completion time of a job $B \in {\mathcal B}$ is defined to be $T(\sigma,B) := \sigma(B) + p(B,h(B))$ and the makespan of $\sigma$ is then $T(\sigma) := \max_{B \in {\mathcal B}} T(\sigma,B)$.
\section{A 3/2-Approximation Algorithm} \label{sec:greedy}

We warm up showing that an application of the classic \textit{First Fit Decreasing (FFD)} scheduling technique leads to a $3/2$-approximation for the base and additional fee job scheduling problem.
In this algorithm, we prioritize the usage of machines with larger capacities and the scheduling of jobs with larger processing times.

In the main loop, the algorithm considers the machines in non-increasing order of their capacities.
When a machine $i \in M$ is considered, the algorithm schedules the unscheduled jobs in non-increasing order of their processing times onto $i$ until its capacity is fully-utilized or all jobs have been scheduled.
When all the machines are considered, the algorithm schedules the remaining unscheduled jobs arbitrarily, if there is any.
See Algorithm~\ref{alg:greedy} for a detailed description of this algorithm.

\begin{algorithm}
\caption{FFD-Approx$(\Pi = (J, M, p, c))$ }   \label{alg:greedy}
\begin{algorithmic}[1]
    \For {each machine $i \in M$ in non-increasing order of capacities}
        \While{the capacity of $i$ is not fully-utilized and there exists an unassigned job}
            \State Pick the unassigned job $j$ with the largest processing time.
            \State Set $\sigma(j) \gets i$.
        \EndWhile
    \EndFor
    \State Assign all the unassigned jobs arbitrarily, if there is any. \label{algo-greedy-last-step}
    \State \Return $\sigma$
\end{algorithmic}
\end{algorithm}

In the following, we prove the $3/2$-approximation guarantee of Algorithm~\ref{alg:greedy}.
Let $\sigma$ be the output of Algorithm~\ref{alg:greedy} and $\sigma^*$ be an optimal schedule. 
Define 
$$M_a := \left\{ \; i \in M  \; \colon \;  \lvert \sigma^{-1}(i) \rvert = 1 \; \text{ and } \; \sum_{j\in\sigma^{-1}(i)}p_j \ge c_i \; \right\}
\quad \text{and} \quad
M_b := M \setminus M_a.$$
Consider the machines in $M_b$.
The following lemma suggests that it suffices to consider the non-degenerating case for which $T(\sigma,i) < 2\cdot c_i$ for all $i \in M_b$.
We provide the proof in Section~\ref{appendix-sec-lemsetBunderfit} in the appendix for further reference.
\begin{restatable}{lemma}{lemsetBunderfit}{\normalfont (Section~\ref{appendix-sec-lemsetBunderfit}).}
\label{lem:set_B_underfit}
    If $T(\sigma,i) \ge 2\cdot c_i$ for some $i \in M_b$, then $T(\sigma) = T(\sigma^*)$.
\end{restatable}
The following main technical lemma proves the optimality of the scheduling decisions for the jobs that are assigned to $M_a$. 
We provide the proof in Section~\ref{appendix-sec-lemgreedysetA} for further reference.
\begin{restatable}{lemma}{lemgreedysetA}{\normalfont (Section~\ref{appendix-sec-lemgreedysetA}).} \label{lem:greedy_set_A}
For any schedule $\sigma'$ of the jobs, there always exists $\sigma''$ such that $T(\sigma'')  \le  T(\sigma')$ and 
$\sigma''^{-1}(i) = \sigma^{-1} (i)$ for all $i \in M_a$.
\end{restatable}
By~\autoref{lem:set_B_underfit} and~\autoref{lem:greedy_set_A}, we obtain the following theorem.
\begin{theorem} \label{theo:1.5-approx} 
    Algorithm~\ref{alg:greedy} computes a 3/2-approximation solution in $O((n+m)\log(n+m))$ time, where $n = |J|$ and $m = |M|$.
\end{theorem}
\begin{proof}
By~\autoref{lem:greedy_set_A}, there exists an optimal schedule $\sigma^*$ such that $\sigma^{*-1}(i) = \sigma^{-1}(i)$ for all $i \in M_a$.
If $T(\sigma, i) \ge 2 \cdot c_i$ for some $i \in M_b$, then $T(\sigma) = T(\sigma^*)$ by~\autoref{lem:set_B_underfit} and we are done.
In the following we assume that $T(\sigma, i) < 2 \cdot c_i$ for all $i \in M_b$.
This implies that
\begin{equation}
    \max\left\{ \; \sum_{j\in \sigma^{-1}(i)} p_j - c_i, \; 0 \; \right\} \; \le \;\; \frac{1}{2} \cdot \sum_{j \in \sigma^{-1}(i)} p_j
    \label{eq:1.5-approx-ieq-1}
\end{equation}
holds for all $i \in M_b$.
It follows that
\begin{align*}
T(\sigma) \; 
& = \; \sum_{i\in M_a}T(\sigma, i) \; + \; \sum_{i\in M_b} T(\sigma, i) \\[1pt]
& = \; \left( \; \sum_{i\in M_a}T(\sigma, i) \; + \; \sum_{i\in M_b} c_i \; \right) \; + \; \sum_{i\in M_b} \max\left\{ \; \sum_{j\in \sigma^{-1}(i)}p_j-c_i, \; 0 \; \right\} \\[3pt]
& \leq \; T(\sigma^*) \; + \; \frac{1}{2} \cdot \sum_{i \in M_b} \sum_{j \in \sigma^{-1}(i)} p_j  \; \le \;  \frac{3}{2} \cdot T(\sigma^{*}),
\end{align*}
where in the second last inequality we apply Inequality~(\ref{eq:1.5-approx-ieq-1}) and in the last inequality we use the fact that 
$ \sum_{i \in M_b} \sum_{j \in \sigma^{-1}(i)} p_j \le \sum_{j \in J} p_j \le T(\sigma^{*}).$
\end{proof}
\paragraph*{A Tight Example for FFD-Approx.}
We provide a tight example for the FFD-Approx algorithm.
Consider the instance $(J, M, p, c)$ where $J=\left\{j_1, j_2\right\}$, $M=\left\{i_1, i_2\right\}$, $p_{j_1} = p_{j_2}= 1 - \epsilon$, $c_{i_1} = c_{i_2} =1$, and $0 < \epsilon < 1$ is an arbitrarily real number.
It is clear that an optimal schedule $\sigma^*$ can be obtained by assigning $j_1$ to $i_1$ and $j_2$ to $i_2$, and the working time of $\sigma^*$ is 
$T(\sigma^*) = \max\left\{ p_{j_1},  c_{i_1}  \right\} + \max\left\{ p_{j_2},  c_{i_2}\right\} = 2$.
On the other hand, following the steps of~\autoref{alg:greedy}, we obtain a schedule $\sigma$ in which both $j_1$ and $j_2$ are assigned to $i_1$, with a total working time of $T(\sigma) = (p_{j_1} + p_{j_2}) + c_{i_2} = 3 - 2\epsilon$.
\section{Approximation Schemes} \label{sec:approx-schemes}

We derive components as building blocks to design a PTAS algorithm for the BAF job scheduling problem.
In the following, we first describe the components.
Then we present the main algorithm.
We provide detailed analyses of the components and the algorithm in the subsequent subsections for further references.

\smallskip
Let $\Pi = (J, M, p, c)$ be an instance of interests.
Our first component is a dynamic programming procedure {\sc DP}$(\Pi)$ that takes as input an instance $\Pi$ and computes an optimal schedule.
We state the guarantee of this procedure as the following lemma and provide the proof in Section~\ref{sec:approx-schemes-procedures} for further reference.
\begin{restatable}{lemma}{lemmadp}
\label{lem:dp}
There is a procedure {\sc DP}$(\Pi = (J, M, p, c))$ that computes an optimal schedule for $\Pi$ in time $m\cdot n^{O(k)}$, where $n = |J|$, $m = |M|$, and $k = \lvert \left\{ p_j \right\}_{j \in J} \rvert$ is the number of different job processing times.
\end{restatable}
To speed-up the dynamic programming procedure, a classic technique is to properly reduce the complexity of the configurations it requires.
For this, our second component is a round-down procedure {\sc RD}$(J, p, \epsilon)$ that takes as input a job set $J$, a processing time function for $J$, and an error parameter $\epsilon$, and creates a well-structured processing time function $\hat{p}$.

\smallskip
In particular, let $p_{\min}(J) := \min_{j\in J}\{p_j\}$ denote the minimum job processing time.
For each $j \in J$, define
$$\hat{p}_j \; := \;  p_{\min}(J) \cdot (1+\epsilon)^{k_j}, \quad \text{where } k_j := \left\lfloor \; \log_{1+\epsilon} \left( \frac{p_j}{ \; p_{\min}(J) \; } \right) \; \right\rfloor.$$
It follows that $\hat{p}$ satisfies the following three properties:
\begin{enumerate}
    \item 
        $\hat{p}_j \le p_j \le (1+\epsilon) \cdot \hat{p}_j$ for all $j \in J$.
        
    \item 
        $\hat{p}_{\min}(J) = p_{\min}(J)$. 
        
    \item
        There are at most $k_{J} := \lfloor \log_{1+\epsilon}( {p_{\max}(J)} / {p_{\min}(J)} ) \rfloor + 1$ different types of processing times in $\hat{p}$, where $p_{\max}(J) := \max_{j\in J}\{p_j\}$ is the maximum job processing time.

\end{enumerate}
\noindent
Consider the two instances $\Pi = (J, M, p, c)$ and $\hat{\Pi} = (J, M, \hat{p}, c)$ of BAF job scheduling.
By the first property for $\hat{p}$ in the above, we have 
$T_{\Pi}(\hat{\sigma}) \le (1+\epsilon) \cdot T_{\hat{\Pi}}(\hat{\sigma}) \le (1+\epsilon) \cdot T_{\Pi}(\sigma),$
where $\sigma$ and $\hat{\sigma}$ are optimal schedules for $\Pi$ and $\hat{\Pi}$, respectively.
Combining the two components directly yields the following ``almost-QPTAS'' algorithm in~\autoref{alg:almost-qptas}.

\begin{algorithm}
\caption{Almost-QPTAS$(\Pi = (J, M, p, c), \epsilon)$ }   \label{alg:almost-qptas}
\begin{algorithmic}[1]
    \State $\hat{p} \gets \text{\sc RD}(J, p, \epsilon)$.
    \State $\sigma \gets \text{\sc DP}(\hat{\Pi} = (J, M, \hat{p}, c))$.
    \State \Return $\sigma$
\end{algorithmic}
\end{algorithm}

The following lemma summarizes the guarantee for~\autoref{alg:almost-qptas}.
We provide the details in Section~\ref{sec:approx-schemes-procedures} for further reference.
\begin{restatable}{lemma}{lemdpfunction}
\label{lem:dp_function}
\autoref{alg:almost-qptas} computes a $(1+\epsilon)$-approximate solution for $\Pi = (J, M, p, c)$ in time $m\cdot n^{O(k_J)}$, where $n = |J|$, $m = |M|$, and $k_J = O(\log_{1+\epsilon}( {p_{\max}(J)} / {p_{\min}(J)} ))$.
\end{restatable}
We note that~\autoref{alg:almost-qptas} can be modified to yield a quasi-polynomial time approximation scheme (QPTAS) by discarding jobs with processing time smaller than $\epsilon / n$ from the DP procedure and scheduling them separately afterwards.
However, this approach does not further improve to a polynomial-time approximation scheme (PTAS). 
To obtain a PTAS, we employ an approach that discards substantially more jobs while creating proper gaps between jobs of different magnitudes in size so that they can be handled separately.
\paragraph*{A PTAS for the Base and Additional Fee Job Scheduling}
In the following we use~\autoref{alg:almost-qptas} as a building block to design a PTAS for the BAF job scheduling problem.
Let $\Psi = (J,M,p,c)$ be the input instance and $0 < \epsilon \le 1/2$ be the target error parameter.
Without loss of generality, we assume that $\max_{j\in J}\{p_j\}=1$.

\smallskip
Let $K := \lceil \frac{1}{\epsilon} \rceil$.
For any $i \ge 0$, define the \textit{layer} 
    $I(i) \; := \; \left\{ \; j\in J \; \colon \; \epsilon^{i+1} < p_j \leq \epsilon^i \; \right\}$
to be the set of jobs whose processing times fall within the $i$th  interval $(\epsilon^{i+1}, \epsilon^i]$.
We partition the family $\{I(0), I(1),\dots\}$ into $K$ equivalence classes $G_0,\dots G_{K-1}$, where $I(i)$ and $I(j)$ belong to the same class if $i\equiv j\pmod{K}$.

\smallskip
Our PTAS aims at creating ``light-enough gaps'' by discarding certain layers. Then, every maximal subfamily of consecutive layers form a \textit{block} to be solved independently. 
In particular, let $t := \arg\min_{0 \le i < K} \sum_{j \in G_i} p_j$ and 
observe that the choice of $t$ satisfies that
$$\sum_{j\in G_t}p_j \; \leq \; (1/K)\cdot \sum_{j\in J}p_j \; \leq \; \epsilon\cdot\sum_{j\in J} p_j.$$ 
We treat each layer in $G_t$ as a gap and discard all of them. Then the set of remaining jobs $J\setminus G_t$ is partitioned into blocks $B_0^{(t)}, B_1^{(t)},\dots$, where 
$$B^{(t)}_0 \; := \; \bigcup_{0\le i' < t} I(i') 
\qquad \text{and} \qquad
B^{(t)}_i \; := \; \bigcup_{0 < i' < K} I \left( (i-1) \cdot K+i'+t \right) \text{ for all } i > 0.$$
We note that, while $B^{(t)}_i$ is defined for all $i$, at most $|J|$ of them can be nonempty.

\begin{algorithm}
\caption{PTAS-Approx$(\Psi = (J, M, p, c), \epsilon)$} \label{alg:ptas-new}
\begin{algorithmic}[1]
    \State $t := \arg\min_{0 \le i < K} \sum_{j \in G_i} p_j$.
    \State $\sigma \gets \emptyset$.
    \For{each $i = 0,1,2,\ldots$ in order such that $B^{(t)}_i \neq \emptyset$}
        \State For each $q \in M$, define $\hat{c}^{(i)}_q := \max\{ c_q - \sum_{j \in \sigma^{-1}(q)} \hat{p}_j, \; 0 \}$.
        \State $\sigma_i \gets \text{Almost-QPTAS}(\Pi_i = (B^{(t)}_i, M, p, \hat{c}^{(i)}), \epsilon)$.  \quad \texttt{// Apply~\autoref{alg:almost-qptas} on $\Pi_i$.}
        \State Set $\sigma(j) \gets \sigma_i(j)$ for all $j \in B^{(t)}_i$. 
    \EndFor
    \State Pick any $i\in M$ and set $\sigma(j) \gets i$ for all $j \in G_t$.
    \State \textbf{return} $\sigma$
\end{algorithmic}
\end{algorithm}

Our algorithm maintains a working schedule $\sigma$, which is initially empty, and considers the nonempty sets among $\{B^{(t)}_i\}_i$ in the order of their indexes.
In each iteration, say, when $B^{(t)}_i$ is considered, the algorithm applies the procedure {Almost-QPTAS}$(\Pi_i = (B^{(t)}_i, M, p, \hat{c}^{(i)}))$ above to obtain a schedule $\sigma_i$ for $\Pi_i$, where $\hat{c}^{(i)}_q$ for each $q \in M$ is defined as 
$$\hat{c}^{(i)}_q \; := \; \max\left\{ \; c_q - \sum_{j \in \sigma^{-1}(q)} \hat{p}_j, \; 0 \; \right\},$$ 
namely, $\hat{c}^{(i)}$ is the residual capacity function subject to the current working schedule.
Then the algorithm merges the assignments in $\sigma_i$ into the working schedule $\sigma$ by setting $\sigma(j)$ to be $\sigma_i(j)$ for all $j \in B^{(t)}_i$ and repeats until all nonempty sets are considered.

\subsection{Analysis of \autoref{alg:ptas-new}}
In this section, we prove the following theorem.
\begin{theorem} \label{theo:ptas}
For any $0<\epsilon\le 1/2$,~\autoref{alg:ptas-new} computes $(1+5\epsilon)$-approximation solution for the BAF Job Scheduling problem in $m \cdot n^{O(\frac{1}{\epsilon} \log_{1+\epsilon}{(\frac{1}{\epsilon})})}$ time, where $n = |J|$ and $m = |M|$.
\end{theorem}
First we analyze the time complexity of \autoref{alg:ptas-new}.
For each block $B^{(t)}_i$, we have
$$\frac{ \; p_{\max}(B^{(t)}_i) \; }{p_{\min}(B^{(t)}_i)} \; \le \; \frac{\epsilon^{(i-1) \cdot K+t+1}}{\epsilon^{i\cdot K+t}} = \left( \frac{1}{\epsilon} \right)^{K-1} \leq \left( \frac{1}{\epsilon} \right)^{\frac{1}{\epsilon}}.$$
Note that this holds for all $i \ge 0$.
By \autoref{lem:dp_function}, the time complexity of \autoref{alg:ptas-new} is
$$
\sum_{i\geq0}{ \; |M|\cdot|B^{(t)}_i| ^ {O( \frac{1}{\epsilon} \cdot \log_{1+\epsilon} {\left(\frac{1}{\epsilon} \right)})}} \; = \; m \cdot n ^ {O( \frac{1}{\epsilon} \cdot \log_{1+\epsilon} {\left(\frac{1}{\epsilon} \right)})}.
$$
\smallskip
In the rest of this section we analyze the approximation ratio of~\autoref{alg:ptas-new}.
Let $\sigma^*$ be an optimal schedule for the input instance $\Psi = (J,M,p,c)$.
Consider the iterations of~\autoref{alg:ptas-new} and recall that $\sigma_i$ is the schedule the algorithm computes for the instance $\Pi_i = (B^{(t)}_i, M, p, \hat{c}^{(i)})$ in the $i$th iteration.

\smallskip
For any $i \ge 0$, define the job-accumulated instance 
$\hat{\Psi}_i := \left( \; \bigcup_{0\le k \le i} B^{(t)}_k, M, \hat{p}, c \; \right), $
where $\hat{p}$ denotes the rounded-down processing time function computed for $\{ B^{(t)}_k \}_{0 \le k \le i}$ up to the $i$th iteration.
Let $\hat{\sigma}_i := \bigcup_{0 \le k \le i} \sigma_i$ denote the working schedule the algorithm maintains at the end of the $i$th iteration.
It follows that $\hat{\sigma}_i$ is a feasible solution for $\hat{\Psi}_i$.
Consider in particular the instance $$\hat{\Psi} \; := \; \hat{\Psi}_\infty \; = \; \left( \; J \setminus G_t, \; M, \; \hat{p}, \; c \; \right)$$
and the schedule $\hat{\sigma} := \hat{\sigma}_\infty$ for $\hat{\Psi}$.
Let $\hat{\sigma}^*$ be the restriction of $\sigma^*$ to the job set $J \setminus G_t$, i.e., 
$$\hat{\sigma}^*(j) \; := \; \sigma^*(j) \quad \text{for all $j \in J\setminus G_t$.}$$
Then $\hat{\sigma}^*$ is a feasible schedule for $\hat{\Psi}$.
Since $\hat{\sigma}^*$ only schedules a subset of jobs of that scheduled in $\sigma^*$, it follows that $T(\hat{\sigma}^*) \le T(\sigma^*)$ with respect to both $\Psi$ and $\hat{\Psi}$.

\smallskip
To prove the approximation guarantee of $\sigma$ for $\Psi$, we first bound the total working time of $\hat{\sigma}$ in terms of that of $\hat{\sigma}^*$ with the following lemma.
\begin{lemma} \label{lem:ptas_main_alg_approx}
$T_{\hat{\Psi}}(\hat{\sigma}) \leq 
\left( (1+ {\epsilon})/({1-\epsilon^K}) \right) \cdot T_{\hat{\Psi}}(\hat{\sigma}^*)$, where $T_{\hat{\Psi}}$ denotes the total working time of the given schedule with respect to the instance $\hat{\Psi}$.
\end{lemma}
To prove~\autoref{lem:ptas_main_alg_approx}, we proceed by considering the iterations of the algorithm.
For each $i = 0, 1, 2, \ldots$ in order, we define an intermediate schedule $\hat{\sigma}^*_i$ from $\hat{\sigma}^*_{i-1}$, where we use $\hat{\sigma}^*_{-1} := \hat{\sigma}^*$ to denote the initial (optimal) schedule, so as to maintain the following two invariant properties during the process:
\begin{enumerate}
    \item 
        $\hat{\sigma}^*_i$ aligns with the assignments made in $\hat{\sigma}_i$, i.e., $\hat{\sigma}^*_i(j) = \hat{\sigma}_i(j)$ for all $j \in \bigcup_{0\le k\le i} B^{(t)}_k$.

    \smallskip
    
    \item 
        The increase of total working times from $\hat{\sigma}^*_{i-1}$ to $\hat{\sigma}^*_i$ is bounded by $\epsilon^{i\cdot K + t+1} \cdot \left| M^*_i \right|$, 
        where
        $$M^*_i \; := \; \left\{ \; q \in M \; \colon \; T_{\hat{\Psi}}(\hat{\sigma}^*, q) \; \ge \; \epsilon^{i\cdot K + t} \; \right\}$$
        denotes the set of machines that contribute a decent amount of workload in $\hat{\sigma}^*$.
        Formally,
        $T_{\hat{\Psi}}(\hat{\sigma}^*_i) \; \le \; T_{\hat{\Psi}}(\hat{\sigma}^*_{i-1}) \; + \; \epsilon^{i\cdot K + t + 1} \cdot \left| M^*_i \right|.$

\end{enumerate}
\medskip
Let $i \ge 0$ be an index of interests.
In the following we describe the construction of $\hat{\sigma}^*_i$ and prove that it satisfies the above two invariant properties.
Let
$$M^{\text{rel}}_i \; := \; \left\{ \; \hat{\sigma}^*_{i-1}(j), \; \hat{\sigma}_i(j) \; \colon \; j \in B^{(t)}_i \; \right\}$$
denote the set of \emph{relevant machines} on which the jobs in $B^{(t)}_i$ are scheduled in $\hat{\sigma}^*_{i-1}$ and $\hat{\sigma}_i$, respectively,
and 
$$J^{(> i)} \;\; := \;\; \bigcup_{k > i} B^{(t)}_k \; \cap \; \bigcup_{\ell \in M^{\text{rel}}_i} \hat{\sigma}^{*-1}_{i-1}(\ell)$$
to be the set of ``small jobs'' collected from $B^{(t)}_{i+1}, B^{(t)}_{i+2}, \ldots$ that are assigned by $\hat{\sigma}^*_{i-1}$ to the relevant machines.
We define the first intermediate schedule $\hat{\sigma}'_i$ as 
$$
\hat{\sigma}'_i(j) \; := 
\begin{cases}
    \; \hat{\sigma}^*_{i-1}(j),  & \text{if $j \in J \setminus \left( \; B^{(t)}_i \cup J^{(> i)} \; \right)$}, 
    \\[3pt]
    \; \hat{\sigma}_i(j), & \text{if $j \in B^{(t)}_i$}.
\end{cases} 
$$
Intuitively, $\hat{\sigma}'_i$ is obtained from $\hat{\sigma}^*_{i-1}$ by reassigning jobs in $B^{(t)}_i$ according to the schedule $\hat{\sigma}_i$ and discarding the assignments for ``the small jobs'' in $J^{(>i)}$.

\smallskip
The following lemma, provided in Section~\ref{appendix-sec-lemptasdpoptimality}, comes from the optimality of $\sigma_i$, and hence also $\hat{\sigma}_i$, for the instance $\hat{\Pi}_i := (B^{(t)}_i, M, \hat{p}, \hat{c}^{(i)})$ defined in~\autoref{alg:almost-qptas} and $\hat{\sigma}'_i$.
\begin{restatable}{lemma}{lemptasdpoptimality}{\normalfont (Section~\ref{appendix-sec-lemptasdpoptimality}).}
\label{lem:ptas-dp-optimality}
    $T_{\hat{\Psi}}(\hat{\sigma}'_i) \; \le \; T_{\hat{\Psi}}(\hat{\sigma}^*_{i-1})$.
\end{restatable}

\begin{figure}[tp]
Procedure {\sc ScheduleSmall}$(\hat{\sigma}'_i, M^{\text{rel}}_i, J^{(>i)})$
\begin{algorithmic}[1]
    \State $\hat{\sigma}^*_i \gets \hat{\sigma}'_i$.
    \For{each $j \in J^{(>i)}$ in any order}
        \If{some $i \in M^{\text{rel}}_i$ is not yet fully-utilized}
            \State Pick one such $i$ and set $\hat{\sigma}^*_i(j) \gets i$.
        \Else
            \State Pick an arbitrary $i \in M^{\text{rel}}_i$ and set $\hat{\sigma}^*_i(j) \gets i$.
        \EndIf
    \EndFor
    \State \Return $\hat{\sigma}^*_i$.
\end{algorithmic}
\caption{A procedure for scheduling the remaining small jobs in $J^{(>i)}$ on relevant machines.}
\label{alg:ptas-proof-assign-small-jobs}
\end{figure}

Provided $\hat{\sigma}'_i$, we obtain the intermediate schedule $\hat{\sigma}^*_i$ by scheduling the small jobs in $J^{(>i)}$ back to the relevant machines with the procedure {\sc ScheduleSmall} in~\autoref{alg:ptas-proof-assign-small-jobs}.
Also refer to Figure~\ref{fig-ptas-iter} for an illustration.
It is clear that $\hat{\sigma}^*_i$ satisfies the first invariant condition stated in the above.
To establish the second condition, we need the following lemma, which is proved by the definition of relevant machines and the design of the procedure.
We provide the proof in Section~\ref{appendix-sec-lemrelevantmachinelowerbound}.

\begin{figure*}[h]
\centering
\includegraphics[scale=0.86]{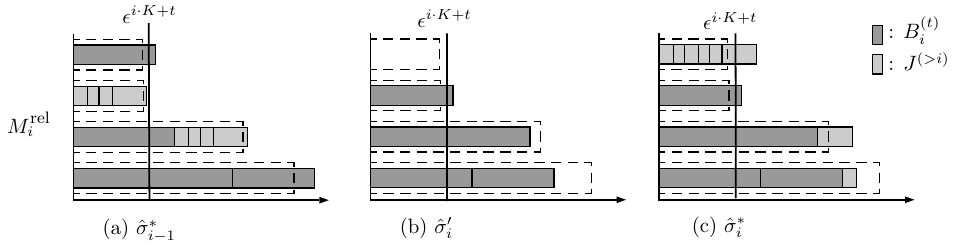}
\caption{An illustration on the construction of $\hat{\sigma}^*_i$ from $\hat{\sigma}^*_{i-1}$ and $\hat{\sigma}'_i$, in particular the schedules on the relevant machines in $M^{\text{rel}}_i$.}
\label{fig-ptas-iter}
\end{figure*}

\begin{restatable}{lemma}{lemrelevantmachinelowerbound}{\normalfont (Section~\ref{appendix-sec-lemrelevantmachinelowerbound}).}
\label{lem:relevant-machine-lower-bound}
For any $i \ge 0$ and any $q \in M$, if $q \in M^{\text{rel}}_i$ is relevant, then there exists $j \in \bigcup_{0\le k\le i} B^{(t)}_k$ such that $\hat{\sigma}'_i(j) = q$ or $\hat{\sigma}^*(j) = q$.
That is, at least one of $\hat{\sigma}'_i$ or $\hat{\sigma}^*$ must schedule some job in $\bigcup_{0\le k\le i} B^{(t)}_k$ on $q$.
\end{restatable}

The following main technical lemma establishes the second invariant condition for $\hat{\sigma}^*$.
\begin{lemma} \label{lem:ptas-proof-intermediate-schedule-bound}
$T_{\hat{\Psi}}(\hat{\sigma}^*_i) \; \le \; T_{\hat{\Psi}}(\hat{\sigma}^*_{i-1}) \; + \; \epsilon^{i\cdot K + t + 1} \cdot \left| M^*_i \right|$.
\end{lemma}
\begin{proof}
Consider the utilization of relevant machines in $M^{\text{rel}}_i$ in the schedule $\hat{\sigma}^*_i$.
If all the machines in $M^{\text{rel}}_i$ are fully-utilized, then we have
$$\sum_{q \in M^{\text{rel}}_i} T_{\hat{\Psi}}(\hat{\sigma}^*_i, q) \; = \; \sum_{q \in M^{\text{rel}}_i} \; \sum_{j \in \hat{\sigma}^{*-1}_i(q)} \hat{p}_j \;\; \le \;\; \sum_{q \in M^{\text{rel}}_i} T_{\hat{\Psi}}(\hat{\sigma}^*_{i-1}, q),$$
where the last inequality follows from the fact that $\hat{\sigma}^*_i$ and $\hat{\sigma}^*_{i-1}$ schedule the same set of jobs on the relevant machines and also the definition of workload function.
Hence it follows that $T_{\hat{\Psi}}(\hat{\sigma}^*_i) \le T_{\hat{\Psi}}(\hat{\sigma}^*_{i-1})$ and we are done.

\smallskip
In the following, we assume that some $q \in M^{\text{rel}}_i$ is not fully-utilized in $\hat{\sigma}^*_i$.
To compare the working times of $\hat{\sigma}^*_i$ and $\hat{\sigma}^*_{i-1}$, we first compare that of $\hat{\sigma}'_i$ and $\hat{\sigma}^*_i$.

\smallskip
Consider any relevant machine $q \in M^{\text{rel}}_i$ and the following three cases regarding its utilizations in the two schedules.
\begin{itemize}
    \item 
        If $\sum_{j\in \hat{\sigma}^{* -1}_{i}(q)} \hat{p}_j < c_q$, i.e., $q$ is not fully-utilized in $\hat{\sigma}^*_i$, then so is it in $\hat{\sigma}'_i$. It follows that 
        $$T_{\hat{\Psi}}(\hat{\sigma}^*_i, q) = T_{\hat{\Psi}}(\hat{\sigma}'_i, q).$$

    \smallskip
    
    \item 
        If $\sum_{j\in \hat{\sigma}'^{-1}_{i}(q)} \hat{p}_j \ge c_q$, then~\autoref{alg:ptas-proof-assign-small-jobs} schedules no further job on $q$ since at least one not-fully-utilized relevant machine exists during the process. Hence 
        $$T_{\hat{\Psi}}(\hat{\sigma}^*_i, q) = T_{\hat{\Psi}}(\hat{\sigma}'_i, q).$$

    \smallskip
    
    \item 
        Lastly, if $\sum_{j\in \hat{\sigma}^{* -1}_{i}(q)} \hat{p}_j \ge c_q$ and $\sum_{j\in \hat{\sigma}'^{-1}_{i}(q)} \hat{p}_j < c_q$, then the design of~\autoref{alg:ptas-proof-assign-small-jobs} ensures that 
        $$T_{\hat{\Psi}}(\hat{\sigma}^*_i, q) \; \le \; \sum_{j\in \hat{\sigma}^{* -1}_{i}(q)} \hat{p}_j \; \le \; c_q \; + \; p_{\max}(B^{(t)}_{i+1}) \; \le \; T_{\hat{\Psi}}(\hat{\sigma}'_i, q) + \epsilon^{i\cdot K + t + 1},$$ 
        where the last inequality follows from $T_{\hat{\Psi}}(\hat{\sigma}'_i, q) \ge c_q$ and $p_{\max}(B^{(t)}_{i+1}) \le \epsilon^{i\cdot K + t + 1}$.

\end{itemize}
\noindent
From the above three cases and~\autoref{lem:ptas-dp-optimality}, we have 
$$ T_{\hat{\Psi}}(\hat{\sigma}^*_i) \; \le \; T_{\hat{\Psi}}(\hat{\sigma}'_i) \; + \; \epsilon^{i\cdot K + t + 1} \cdot \left| M^{\phi}_i \right| \; \le \; T_{\hat{\Psi}}(\hat{\sigma}^*_{i-1}) \; + \; \epsilon^{i\cdot K + t + 1} \cdot \left| M^{\phi}_i \right|, $$
where $M^{\phi}_i := \left\{ \; q \in M^{\text{rel}}_i \; \colon \; \sum_{j\in \hat{\sigma}'^{-1}_{i}(q)} \hat{p}_j < c_q \; \right\}$.
To finish the proof for this lemma, it remains to show that 
$$M^{\phi}_i \subseteq M^*_i.$$
To prove this statement, consider any $q \in M^{\phi}_i$ and we have two cases.
If $c_q \ge \epsilon^{i\cdot K +t}$, then $T_{\hat{\Psi}}(\hat{\sigma}^*, q) \ge c_q \ge \epsilon^{i\cdot K +t}$ and $q \in M^*_i$ in this case.

\smallskip
Now suppose that $c_q < \epsilon^{i\cdot K +t}$.
Since $p_{\min}(B^{(t)}_k) \ge p_{\min}(B^{(t)}_i) \ge \epsilon^{i\cdot K +t}$ for all $0\le k\le i$,
it follows that $\hat{\sigma}'_i$ schedules no job in $\bigcup_{0\le k\le i} B^{(t)}_k$ on $q$.
By Lemma~\ref{lem:relevant-machine-lower-bound}, $\hat{\sigma}^*$ must schedule some job in $\bigcup_{0\le k\le i} B^{(t)}_k$ on $q$, and 
$$T_{\hat{\Psi}}(\hat{\sigma}^*, q) \; \ge \; \hat{p}_{\min}(B^{(t)}_i) \; = \; p_{\min}(B^{(t)}_i) \; \ge \; \epsilon^{i\cdot K+t},$$
where the equality follows from the second property of the rounded-down procedure.
This gives that $q \in M^*_i$ and this lemma follows.
\end{proof}
Now we are ready to prove~\autoref{lem:ptas_main_alg_approx}.
\begin{proof}[Proof of~\autoref{lem:ptas_main_alg_approx}]
Define $M_0 \; := \; \left\{ \; q \in M \; \colon \; \epsilon^{t} \; \le \; T_{\hat{\Psi}}(\hat{\sigma}^*, q) \; \right\}$ and
$$M_i \; := \; \left\{ \; q \in M \; \colon \; \epsilon^{i\cdot K + t} \; \le \; T_{\hat{\Psi}}(\hat{\sigma}^*, q) \; < \; \epsilon^{(i-1)\cdot K + t} \; \right\}$$
for any $i \ge 1$.
Recall that, for any $i \ge 0$, we have
$$M^*_i \; := \; \left\{ \; q \in M \; \colon \; T_{\hat{\Psi}}(\hat{\sigma}^*, q) \; \ge \; \epsilon^{i\cdot K + t} \; \right\}.$$
By~\autoref{lem:ptas-proof-intermediate-schedule-bound}, we have
\begin{align*}
T_{\hat{\Psi}}( \hat{\sigma} ) \;\; = \;\; T_{\hat{\Psi}}( \hat{\sigma}_\infty ) \;\; 
\le & \;\; T_{\hat{\Psi}}(\hat{\sigma}^*) \; + \; \sum_{i \ge 0} \; \epsilon^{i\cdot K + t + 1} \cdot \left| M^*_i \right| \\[3pt]
= & \;\; T_{\hat{\Psi}}(\hat{\sigma}^*) \; + \; \sum_{i \ge 0} \; \; \left| M_i \right|  \; \cdot \; \sum_{i' \ge i} \epsilon^{i'\cdot K + t + 1},
\end{align*}
where the last equality follows from the definitions of $M_i$ and $M^*_i$.
Hence, 
\begin{align*}
T_{\hat{\Psi}}( \hat{\sigma} ) \;\; 
\le & \;\;\;  T_{\hat{\Psi}}(\hat{\sigma}^*) \;\; + \;\; \sum_{i \ge 0} \; \; \left| M_i \right|  \cdot \epsilon^{i\cdot K + t} \cdot \sum_{i' \ge 0} \epsilon^{i'\cdot K + 1} \\[3pt]
\le & \;\;\; T_{\hat{\Psi}}(\hat{\sigma}^*) \;\; + \;\; \frac{\epsilon}{ \; 1-\epsilon^K \; } \cdot \sum_{i \in M} T_{\hat{\Psi}}(\hat{\sigma}^*, i)
\;\; \leq \;\; \frac{1+\epsilon}{ \; 1-\epsilon^K \; } \cdot T_{\hat{\Psi}}(\hat{\sigma}^*).
\end{align*}
This proves the lemma.
\end{proof}
\smallskip
Consider the input instance $\Psi = (J,M,p,c)$.
Recall that $\sigma^*$ is an optimal schedule for $\Psi$ and $\sigma$ is the schedule~\autoref{alg:ptas-new} computes for $\Psi$.
We have 
\begin{align}
T_{\Psi}(\sigma) \; \le \; T_{\Psi}(\hat{\sigma}) \; + \; \sum_{j \in G_t} p_j \; 
\le & \;\; T_{\Psi}(\hat{\sigma}) \; + \; \epsilon \cdot \sum_{j \in J} p_j  \notag \\[2pt]
\le & \;\; T_{\Psi}(\hat{\sigma}) \; + \; \epsilon \cdot T_{\Psi}(\sigma^*)
\label{eq-ptas-final-1}
\end{align}
by the way $t$ is chosen and that $T_\Psi(\sigma^*) \ge \sum_{j \in J} p_j$.
For $T_{\Psi}(\hat{\sigma})$, we have 
\begin{align*}
T_{\Psi}(\hat{\sigma}) \;\; = \;\; 
\sum_{i \in M} \max\left\{ \; \sum_{j \in \hat{\sigma}^{-1}(i)}p_j , \;\; c_i \; \right\} \;\; 
\leq & \;\; \sum_{i \in M} \max\left\{ \; \sum_{j \in \hat{\sigma}^{-1}(i)}(1+\epsilon) \cdot \hat{p}_j , \;\; c_i \; \right\}  \\[6pt]
\leq & \;\; (1+\epsilon) \cdot T_{\hat{\Psi}}(\hat{\sigma}),
\end{align*}
where the first inequality follows from the first property of the rounded-down processing time function $\hat{p}$.
Apply~\autoref{lem:ptas_main_alg_approx} and we have 
\begin{align}
T_{\Psi}(\hat{\sigma}) \; \leq \; (1+\epsilon) \cdot \frac{1+\epsilon}{ \; 1-\epsilon^K \; } \cdot T_{\hat{\Psi}}(\hat{\sigma}^*) \; 
\le & \;\; \frac{(1+\epsilon)^2}{ \; 1-\epsilon^K \; } \cdot T_\Psi(\hat{\sigma}^*) \;
\le \; \frac{(1+\epsilon)^2}{ \; 1-\epsilon^K \; } \cdot T_\Psi(\sigma^*),
\label{eq-ptas-final-2}
\end{align}
where in the second last inequality we use the fact that $\hat{p}_j \le p_j$ for all $j \in J$ and the last inequality follows from the fact that $T_\Psi(\hat{\sigma}^*) \le T_\Psi(\sigma^*)$.
By~(\ref{eq-ptas-final-1}) and~(\ref{eq-ptas-final-2}), we obtain that
$$ T_{\Psi}(\sigma) \; \le \; \left( \; \frac{(1+\epsilon)^2}{1-\epsilon^K} \; + \; \epsilon \; \right) \cdot T_{\Psi}(\sigma^*) \; \le \; 
(1+5\epsilon) \cdot T_{\Psi}(\sigma^*)$$
where the last inequality follows from the setting that $0 < \epsilon \le 1/2$. 
This finishes the proof for~\autoref{theo:ptas}.
\section{Towards an Asymptotic L/2-Approximation for Non-Preemptive Mixed-Criticality Scheduling} \label{sec:multi-case}

In this section we consider the non-preemptive mixed-criticality match-up scheduling problem. 
Let $\Psi = (L, {\mathcal B}, h, p)$ be an instance of this problem.
We show that the PTAS presented in the previous section for the BAF job scheduling problem can be applied to yield an asymptotic $(L(1+\epsilon))/2$-approximation for this problem.

\smallskip
Define $S_\ell:=\{B\in\mathcal{B} \colon h(B)=\ell\}$, where $1\le \ell \le L$, to be the set of F-shaped jobs with a height of $\ell$.
We group every two consecutive sets by defining $P_i = (S_{2i-1}, S_{2i})$ for any $i$ with $1\le i \le \lfloor L/2 \rfloor$.
For any such $i$, we define an instance $\Pi_i = (J^{(i)}, M^{(i)}, p^{(i)}, c^{(i)})$ for BAF scheduling, where
\begin{equation*}
    J^{(i)} := S_{2i-1}, \quad
    M^{(i)} := S_{2i},  \quad
    p^{(i)}_B := p(B, 2i-1) \text{ for all $B \in S_{2i-1}$},
\end{equation*}
and $c^{(i)}_B := p(B,2i) - p(B,2i-1)$ for all $B \in S_{2i}$.
Then the procedure PTAS-Approx$(\Pi_i, \epsilon)$ from~\autoref{alg:ptas-new} is applied to obtain a schedule $\sigma_i$ for $\Pi_i$.
Intuitively, we treat each $P_i$ as a two-layer instance and apply the PTAS algorithm on it.
Provided $\sigma_i$ for all $1\le i \le \lfloor L/2 \rfloor$, we define an assignment function $\hat{\sigma} \colon {\mathcal B} \to {\mathcal B} \cup \{\phi\}$ as 
$$
\hat{\sigma}(B) \; := \; \begin{cases}
    \; \phi, & \text{if $B \in S_{2i}$ for some $1\le i \le \lfloor L/2 \rfloor$,} \\[2pt]
    \; \sigma_i(B), & \text{if $B \in S_{2i-1}$ for some $1\le i \le \lfloor L/2 \rfloor$,} \\[2pt]
    \; \phi, & \text{if $L$ is odd and $B \in S_L$,}
\end{cases}
\quad \text{for any $B \in {\mathcal B}$.}
$$
To output a schedule $\sigma$ for $\Psi$, the algorithm iterate over all $B \in \hat{\sigma}^{-1}(\phi)$ in order.
For each job $B$ considered, let $t_B$ denote the job completion time of the current schedule $\sigma$.
The algorithm first schedules $B$ at $t_B$ and then all jobs in $\hat{\sigma}^{-1}(B)$ in any order, starting from $t_B + p(B,h(B)-1)$.
See also Figure~\ref{fig-mixed-crit-L2-1} for an illustration.

\begin{figure*}[h]
\centering
\includegraphics[scale=0.9]{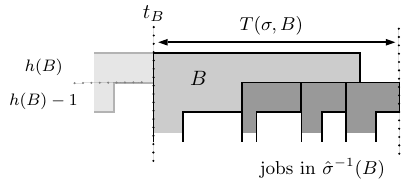}
\caption{An illustration on the schedule $\sigma$ obtained from $\hat{\sigma}$.}
\label{fig-mixed-crit-L2-1}
\end{figure*}

\noindent
It is clear that the job completion time of $\sigma$ is $T(\sigma) \; \le \; \sum_{B \in \hat{\sigma}^{-1}(\phi)} T(\sigma,B)$, where
$$T(\sigma,B) \; := \; \max\left\{ \; p(B,h(B)), \;\;\; p(B,h(B)-1)) + \; \sum_{B' \in \hat{\sigma}^{-1}(B)} p(B',h(B')) \; \right\}.$$
We conclude this section with the following theorem.
\begin{theorem}
    We can compute a $\lceil L/2 \rceil (1+O(\epsilon))$-approximation for the input instance $\Psi = ({\mathcal B}, h, p)$ with and $L := \max_{B \in \mathcal{B}}h(B)$ in time $L|\mathcal{B}|^{O(\frac{1}{\epsilon}\log_{1+\epsilon}{\frac{1}{\epsilon}})}$.
\end{theorem}
\begin{proof}
Let $\sigma^*$ be an optimal schedule of instance $\Psi$.
It is clear that, for any $1\le i \le \lfloor L/2 \rfloor$, the completion time of the jobs in $S_{2i-1}$ and $S_{2i}$ in $\sigma^*$ is a lower bound on $T(\sigma^*)$.
Hence, by~\autoref{theo:ptas}, we have $\sum_{B \in S_{2i}} T(\sigma,B) \leq (1+O(\epsilon))\cdot T(\sigma^*)$.
Sum over all such $i$ and $\sum_{B \in S_{L}} T(\sigma,B)$ if $L$ is odd, we obtain the claimed approximation guarantee.
\end{proof}
\begin{appendix}

\section{Proofs of Technical Lemmas}

\label{sec-appendix-proof}

%


%

\subsection{Proof of~\autoref{lem:set_B_underfit}}

\label{appendix-sec-lemsetBunderfit}

\lemsetBunderfit*

\begin{proof}
Let $i$ be a machine in $M_b$ with $T(\sigma, i) \ge 2 \cdot c_i$.
Then it follows that $\lvert \sigma^{-1}(i) \rvert \ge 2$.
Let $j$ be the last job~\autoref{alg:greedy} schedules on $i$.
By the design of the algorithm, we have $p_{j'} \ge p_{j}$ for all $j' \in \sigma^{-1}(i) \setminus \{j\}$.

\smallskip

Consider the workload $\sum_{j' \in \sigma^{-1}(i)\setminus \left\{j\right\}} p_{j'}$.
If $\sum_{ j' \in \sigma^{-1}(i)\setminus \left\{j\right\}} p_{j'} < c_i$, then
$$p_{j} \; = \; T(\sigma, i) \; - \sum_{ j' \in \sigma^{-1}(i)\setminus \left\{j\right\}} p_{j'} \; > \; 2 \cdot c_i \; - \; c_i \; = \; c_i,$$
which contradicts the property that $p_{j'} \ge p_{j}$ for all $j' \in \sigma^{-1}(i) \setminus \{j\}$.
Hence it must be the case that $\sum_{j' \in \sigma^{-1}(i)\setminus \left\{j\right\}} p_{j'} \ge c_i$.

\smallskip

This means that $j$ must be scheduled by Step~\ref{algo-greedy-last-step} of the~\autoref{alg:greedy}, and $\sum_{ j' \in \sigma^{-1}(i')} p_{j'} \ge  c_{i'}$ for all $i' \in M$.
This gives that 
$$ T(\sigma) \; = \; \sum_{i' \in M} \max\left\{ \sum_{ j' \in \sigma^{-1}(i')} p_{j'}, \; c_{i'} \right\} \; = \; \sum_{j \in J} p_j \; \le \; T(\sigma^*).$$
Since $\sigma^*$ is an optimal schedule, this lemma follows.
\end{proof}

%


%

\subsection{Proof of~\autoref{lem:greedy_set_A}}

\label{appendix-sec-lemgreedysetA}

\lemgreedysetA*

\begin{proof}
Let $i_1, i_2, \ldots, i_m$ be the machines in $M$ indexed according to the order they are considered during the execution of~\autoref{alg:greedy}.
To prove this lemma, we consider the machines in $M_a$ one by one according to the order of their indexes and swap the assignments made in $\sigma'$ step by step, if necessary, to obtain the claimed schedule $\sigma''$ with the stated properties.
Let $\hat{\sigma}$ denote the working schedule during the process, which is initially set to $\sigma'$.
For each machine, say, $i_t \in M_a$ considered according to the above order, we will define a new schedule $\hat{\sigma}'$ such that the following three conditions hold.
\begin{enumerate}
    \item 
        $T(\hat{\sigma}') \; \le \; T(\hat{\sigma})$.

    \item 
        $\hat{\sigma}'^{-1}(i_t) = \sigma^{-1}(i_t)$.
        
    \item 
        $\hat{\sigma}'^{-1}(i) = \hat{\sigma}^{-1}(i)$, for all $i \in M_a$ with indexes strictly smaller than $t$.
\end{enumerate}
Then $\hat{\sigma}$ is replaced with $\hat{\sigma}'$ in the next iteration.
Note that, repeating this argument proves this lemma.
\medskip
Let $i_t \in M_a$ be the machine to be considered in the current iteration and $j_t := \sigma^{-1}(i_t)$ denote the job that is scheduled on $i_t$ in $\sigma$.
Depending on whether or not $j_t \in \hat{\sigma}^{-1}(i_t)$, consider the following two cases.
\begin{itemize}
    \item 
        Case (1) : $j_t \in \hat{\sigma}^{-1}(i_t)$. 

        \smallskip

        For this case, we pick $i'$ for the jobs in $\hat{\sigma}^{-1}(i_t) \setminus \{j_t\}$ as follows. 
        If $M_b \neq \emptyset$, then set $i'$ to be an arbitrary machine in $M_b$.
        Otherwise, there must exist a machine in $M_a$ with a larger index according to the invariant condition we maintain.
        Set $i'$ to be this machine.
        \smallskip
        Obtain a new schedule $\hat{\sigma}'$ by setting for $j\in J$
        $$\hat{\sigma}'(j) \; := \begin{cases}
            \;\; i',  & \text{if $j \in \hat{\sigma}^{-1}(i_t) \setminus \{j_t\}$}, \\[1pt]
            \;\; \hat{\sigma}(j), & \text{otherwise}.
        \end{cases} $$
        Since $i_t \in M_a$, we have $T(\hat{\sigma}', i_t) = p_{j_t} = T(\sigma, i_t) \ge c_{i_t}$.
        It follows that the workload decreased at $i_t$ is exactly $\sum_{ j \in \hat{\sigma}^{-1}(i_t) \setminus \{j_t\}} p_j$, namely,
        $$T(\hat{\sigma}, i_t) \; - \; T(\hat{\sigma}', i_t) \;\; = \;\;  \sum_{ j \in \hat{\sigma}^{-1}(i_t) \setminus \{j_t\}} p_j.$$
        On the other hand, 
        $$T(\hat{\sigma}', i') \; - \; T(\hat{\sigma}, i') \;\; \le \;\;  \sum_{ j \in \hat{\sigma}^{-1}(i_t) \setminus \{j_t\}} p_j$$ by the definition of $T$.
        Hence $T(\hat{\sigma}') \le T(\hat{\sigma})$.
        It is straightforward to verify that $\hat{\sigma}'$ satisfies the remaining two requirements.

    \medskip
    
    \item 
        Case (2) : $j_t \notin \hat{\sigma}^{-1}(i_t)$.

        \smallskip
        
        In this case, let $t'$ be the index of the machine on which $\hat{\sigma}$ schedules $j_t$, i.e., $t'$ is the index such that $\hat{\sigma}(j_t) = i_{t'}$.
        We further consider two subcases.
        
        \smallskip
        
        Case (2a) : If $t' > t$, then we have $c_{i_t} \ge c_{i_{t'}}$.
        In this case, we swap $j_t$ and all jobs in $\hat{\sigma}^{-1}(i_{t})$.
        In particular, define a new schedule $\hat{\sigma}'$ by setting
        $$\hat{\sigma}'(j) \; := \begin{cases}
            \;\; i_t,  & \text{if $j = j_t$}, \\[1pt]
            \;\; i_{t'},  & \text{if $j \in \hat{\sigma}^{-1}(i_t)$},\\[1pt]
            \;\; \hat{\sigma}(j), & \text{otherwise},
        \end{cases} \qquad \text{for any $j \in J$}. $$
        Since $\hat{\sigma}$ and $\hat{\sigma}'$ differ only by the assignments they made to $i_t$ and $i_{t'}$, we focus on comparing the working times of these two machines in $\hat{\sigma}$ and $\hat{\sigma}'$. 
        First, for the schedule $\hat{\sigma}$, we have
        \begin{align*}
        T(\hat{\sigma}, i_t) \; + \; T(\hat{\sigma}, i_{t'}) \;\; 
        & = \; \;  \max\left\{ \; \sum_{j\in \hat{\sigma}^{-1}(i_t)}p_j, \;\; c_{i_t} \; \right\} \; + \; \max\left\{ \; \sum_{j\in \hat{\sigma}^{-1}(i_{t'})}p_j, \;\; c_{i_{t'}} \; \right\}  \\[2pt]
        & = \; \;  \max\left\{ \; \sum_{j\in \hat{\sigma}^{-1}(i_t)}p_j, \;\; c_{i_t} \; \right\} \; + \; \sum_{j\in \hat\sigma^{-1}(i_{t'})}p_j  \\[2pt]
        &= \; \; \max\left\{ \; \sum_{j\in \hat{\sigma}^{-1}(i_t)}p_j \; + \; \sum_{j\in \hat\sigma^{-1}(i_{t'})}p_j \;, \;\; c_{i_t} + \; \sum_{j\in \hat\sigma^{-1}(i_{t'})}p_j \; \right\},
        \end{align*}
        where the second equality follows from the fact that $i_t \in M_a$ and hence $p_{j_t} \ge c_{i_t} \ge c_{i_{t'}}$.
        On the other hand, for the new schedule $\hat{\sigma}'$, 
        \begin{align*}
        T(\hat{\sigma}', i_t) \; + \; T(\hat{\sigma}', i_{t'}) \;\; 
        & = \;\; p_{j_t} \; + \; \max \left\{ \; \sum_{j\in \hat{\sigma}^{-1}(i_t)}p_j \; + \; \sum_{j\in \hat\sigma^{-1}(i_{t'})\setminus \left\{j_t\right\}}p_j, \; \; c_{i_{t'}} \; \right\} \\[3pt]
        & = \;\; \max \left\{ \; \sum_{j\in \hat{\sigma}^{-1}(i_t)}p_j \; + \; \sum_{j\in \hat\sigma^{-1}(i_{t'})}p_j, \; \; c_{i_{t'}} \; + \; p_{j_t} \; \right\}. 
        \end{align*}
        Since $c_{i_{t'}} + p_{j_t}\le c_{i_{t}} + \sum_{j\in \hat\sigma^{-1}(i_{t'})}p_j $ always holds, we have 
        $$T(\hat{\sigma}', i_t) \; + \; T(\hat{\sigma}', i_{t'})  \;\; \le \;\; T(\hat{\sigma}, i_t) \; + \; T(\hat{\sigma}, i_{t'})$$ 
        and $T(\hat{\sigma}') \le T(\hat{\sigma})$.
        The remaining two requirements for $\hat{\sigma}'$ are straightforward.
        \bigskip
        Case (2b) : If $t' < t$, then we have $c_{i_{t'}} \ge c_{i_t}$.
        Consider the set of jobs which~\autoref{alg:greedy} has already scheduled before it schedules $j_t$.
        Let $J_t$ denote the set of these jobs.
        Note that, it follows that $p_j \ge p_{j_t}$ for all $j \in J_t$.
        We have the last two subcases to consider.
        \begin{itemize}
            \item 
                Case (2b-i) : For some $j'' \in J_t$, $\hat{\sigma}$ schedules $j''$ on a machine $i_{t''}$ with $t'' \ge t$.
                \smallskip
                For this case, we swap the assignments made for $j''$ and $j_t$.
                Let $\hat{\sigma}''$ be the schedule defined by 
                $$\hat{\sigma}''(j) \; := \begin{cases}
                \;\; i_{t'},  & \text{if $j = j''$}, \\[1pt]
                \;\; i_{t''},  & \text{if $j = j_t$},\\[1pt]
                \;\; \hat{\sigma}(j), & \text{otherwise},
                \end{cases} \quad \text{for any $j \in J$}. $$
                The argument for proving $T(\hat{\sigma}'') \le T(\hat{\sigma})$ is analogous to that used for the case~(2a).
                For $\hat{\sigma}$, we have
                \begin{align*}
                & T(\hat{\sigma}, i_{t'}) \; + \; T(\hat{\sigma}, i_{t''}) \;\; \\[3pt]
                & = \; \;  \max\left\{ \; \sum_{j\in \hat{\sigma}^{-1}(i_{t'})\setminus \left\{j_t\right\}}p_j \; + \; p_{j_t}, \;\; c_{i_{t'}} \; \right\} \; + \; \sum_{j\in \hat{\sigma}^{-1}(i_{t''})\setminus \left\{j''\right\}} p_j \; + \; p_{j''}.
                \end{align*}
                For $\hat{\sigma}''$, we have
                \begin{align*}
                & T(\hat{\sigma}'', i_{t'}) \; + \; T(\hat{\sigma}'', i_{t''}) \;\; \\[3pt]
                &= \; \; \max\left\{ \; \sum_{j\in \hat{\sigma}^{-1}(i_{t'})\setminus \left\{j_t\right\}}p_j \; + \; p_{j''}, \;\; c_{i_{t'}} \; \right\} \; + \; \sum_{j\in \hat{\sigma}^{-1}(i_{t''})\setminus \left\{j''\right\}} p_j \; + \; p_{j_t}.
                \end{align*}
                Since $p_{j_t} \le p_{j''}$, rewriting the above as was done in case~(2a) gives that 
                \begin{align*} 
                & \max\left\{ \; \sum_{j\in \hat{\sigma}^{-1}(i_{t'}) \cup \hat{\sigma}^{-1}(i_{t''})} p_j, \;\;\;\; c_{i_{t'}} + \sum_{j\in \hat{\sigma}^{-1}(i_{t''})\setminus \left\{j''\right\}} p_j \; + \; p_{j_t} \; \right\} \\[6pt]
                & \hspace{1.2cm} 
                \le \;\; \max\left\{ \; \sum_{j\in \hat{\sigma}^{-1}(i_{t'}) \cup \hat{\sigma}^{-1}(i_{t''})} p_j, \;\;\;\; c_{i_{t'}} + \sum_{j\in \hat{\sigma}^{-1}(i_{t''})\setminus \left\{j''\right\}} p_j \; + \; p_{j''} \; \right\},
                \end{align*}
                and $ T(\hat{\sigma}'', i_{t'}) + T(\hat{\sigma}'', i_{t''}) \; \le \; T(\hat{\sigma}, i_{t'}) + T(\hat{\sigma}, i_{t''})$.
                
                \smallskip
                
                We have proved that there exists $\hat{\sigma}''$ such that $T(\hat{\sigma}'') \le T(\hat{\sigma})$ and $\hat{\sigma}''^{-1}(i) = \hat{\sigma}^{-1}(i)$ for all $i \in M_a$ that has an index smaller than $t$.
                Then we are in a situation that satisfies the same pre-condition either as in case~(1) or as in case~(2a), and the same argument reapplies for the existence of $\hat{\sigma}'$ with the stated requirements.
                
                In particular, if $t'' = t$, then the argument in case~(1) applies. Otherwise, we apply the argument from case~(2a).

            \medskip

            \item
                Case (2b-ii) : $\hat{\sigma}$ schedules all the jobs in $J_t$ on machines with indexes smaller than $t$. 
                Consider the behavior of~\autoref{alg:greedy}.
                Since the algorithm schedules $j_t$ on $i_t$, all the machines with indexes smaller than $t$ are fully-utilized after assigning all the jobs in $J_t$.
                For $i_t$, it was empty before $j_t$ was scheduled and becomes fully-utilized after $j_t$ was assigned.
                Hence we have 
                \begin{equation}
                \sum_{1\le k \le t} T(\sigma, i_k) \; = \; \sum_{j \in J_t \cup \{ j_t\}} p_j.  \label{eq-greedy-2bii-1}
                \end{equation}
                For this case, we reschedule in $\hat{\sigma}$ the assignments made for all the jobs in $J_t \cup \{j_t\}$ to align with the assignments made in $\sigma$.
                In particular, define 
                $$\hat{\sigma}''(j) \; := \begin{cases}
                \;\; \sigma(j),  & \text{if $j \in J_t \cup \{j_t\}$}, \\[1pt]
                \;\; \hat{\sigma}(j), & \text{otherwise},
                \end{cases} \enskip\quad \text{for any $j \in J$}. $$
                Let 
                $$J_s \; := \; \bigcup_{1 \le k \le t} \; \hat{\sigma}^{-1}(i_k) \setminus \left\{J_t \cup \left\{j_t\right\} \right\}$$ 
                be the set of jobs containing all jobs other than $J_t \cup \left\{j_t\right\}$ in the first $t$ machines.
                Consider the working time of $i_1, i_2, \ldots, i_t$ in $\hat{\sigma}$ and $\hat{\sigma}''$. 
                For $\hat{\sigma}$, we have
                \begin{align}
                \sum_{1 \le k \le t} T(\hat{\sigma}, i_k) \;\; 
                & = \;\; \sum_{1 \le k \le t} \max \left\{ \sum_{j\in \hat{\sigma}^{-1}(i_k)}p_j, \; \; c_{i_k} \; \right\}  \notag \\[6pt]
                & \ge \;\; \sum_{1 \le k \le t} \; \sum_{j\in \hat{\sigma}^{-1}(i_k)} p_j \;\; = \;\; \sum_{j\in J_t \cup \left\{j_t\right\}} p_j \; + \; \sum_{j\in J_s} p_j,  
                \label{ieq-greedy-2bii-2}
                \end{align}
                where the first inequality follows from the fact that $\max \left\{ \sum_{j\in \hat{\sigma}^{-1}(i_k)}p_j, \; c_{i_k} \right\} \ge \sum_{j\in \hat{\sigma}^{-1}(i_k)} p_j$ for all $k$.
                \smallskip
                As for $\hat{\sigma}''$, we have
                \begin{align}
                \sum_{1 \le k \le t} T(\hat{\sigma}'', i_k) \;\; 
                = \; \sum_{1 \le k \le t} \max \left\{ \sum_{j\in \hat{\sigma}''^{-1}(i_k)}p_j, \; \; c_{i_k} \; \right\} \; = \; \sum_{j\in J_t \cup \left\{j_t\right\}} p_j \; + \; \sum_{j\in J_s} p_j  
                \label{ieq-greedy-2bii-3}
                \end{align}
                where the second equality follows from the definition of $\hat{\sigma}''$ and~(\ref{eq-greedy-2bii-1}).
                By~(\ref{ieq-greedy-2bii-2}) and~(\ref{ieq-greedy-2bii-3}), we have $T(\hat{\sigma}'') \le T(\hat{\sigma})$.
                Now we are in a situation that satisfies the pre-condition of case~(1) and the same argument reapplies.
                
        \end{itemize}
\end{itemize}
\noindent
This proves the lemma.
\end{proof}

\subsection{Proof of~\autoref{lem:dp} and~\autoref{lem:dp_function}}  
\label{sec:approx-schemes-procedures}

\lemmadp*
\begin{proof}
Let $a_1, a_2, \ldots, a_k$ be the processing times of the jobs in $J$ and for any $1\le i \le k$, let $n_i$ be the number of jobs in $J$ that has a processing time of $a_i$, namely, 
$$n_i \; = \; \left|\left\{ \; j \in J \; \colon \; p_j = a_i \; \right\}\right|.$$
Then, a subset $q$ of jobs can be represented by a $k$-tuple $q=(q_1, \cdots, q_k)$ where $q_i \leq n_i$ for all $1 \le i \le k$.
Let $Q := \left\{ \; (q_1, \cdots, q_k) \; \colon \; 0\le q_i\leq n_i \text{ for all $1\le i\le k$} \; \right\}$ be the set of all possible $k$-tuples.
Since $n = \sum_{1\le i\le k} n_i$, it follows that $|Q| \le \binom{n+k}{k} = O(n^k)$.

\smallskip

Label the machines arbitrarily with indexes from $1$ to $m$.
For any $1\le i \le m$ and any $q \in Q$, 
define $W(i,q)$ to be the minimum sum of working times of the first $i$ machines if the job set $q$ is to be scheduled on them.
We have the following recurrence formula for $W(i,q)$ by the optimal substructure.
\begin{equation*}
W(i,q) \;\; = 
\hspace{-4pt}
\min_{ \substack{ \\[1pt]  q' \in Q, \\[3pt]   0\leq q'_\ell \leq q_\ell \text{ for all } 1\le \ell \leq k } } 
\hspace{-12pt} 
W(i-1, q-q') \; + \; \max \left\{ \; \sum_{1\le \ell \le k} q'_\ell \cdot a_\ell, \;\; c_i \; \right\}
\end{equation*}
for any $i \ge 1$ and $q \in Q$,
where we define the boundary values as
\begin{equation*}
W(0,q) \; = \; \begin{cases}
    \;\; 0, & \text{if $q = (0,0,\ldots,0)$,} \\[2pt]
    \;\; \infty, & \text{otherwise}.
\end{cases}
\end{equation*}
It is clear that $W(m, (n_1, n_2, \ldots, n_k))$ gives an optimal solution for $\Pi$ and it takes $m \cdot n^{O(k)}$ time to compute it.
\end{proof}
%

%


\lemdpfunction*

\begin{proof}
By the three property of the round-down procedure {\sc RD}$(J,\epsilon)$ described in the beginning of~\autoref{sec:approx-schemes}, 
we know that 
there are at most $k_J$ different types of job processing times in $\hat{p}$.
Hence, by~\autoref{lem:dp}, the procedure~{\sc DP} computes an optimal schedule $\sigma$ for the instance $\hat{\Pi}=(J,M, \hat{p}, c)$ in time $m \cdot n^{O(k_J)}$.
We have
\begin{align*}
T_{\Pi}(\sigma) \;\; = \;\; \sum_{i \in M} \; \max\left\{ \; \sum_{j \in \sigma^{-1}(i)}p_j , \;\; c_i \; \right\} \; 
\leq & \;\; \sum_{i \in M} \; \max\left\{ \; \sum_{j \in \sigma^{-1}(i)}(1+\epsilon) \cdot \hat{p}_j , \;\; c_i \; \right\}  \\[8pt]
\leq \;\; (1+\epsilon) \cdot T_{\hat{\Pi}}(\sigma) \;\; 
\le & \;\; (1 + \epsilon) \cdot T_{\hat{\Pi}}(\sigma^*) \;\; \le \;\; (1 + \epsilon) \cdot T_{\Pi}( \sigma^* ),
\end{align*} 
where $\sigma^*$ denotes an optimal schedule for the input instance $\Pi=(J,M, p, c)$, the second last inequality follows from the optimality of $\sigma$ for $\hat{\Pi}$, and the last inequality follows from the first property of the round-down procedure {\sc DP}. 
\end{proof}
%


\subsection{Proof of~\autoref{lem:ptas-dp-optimality}}

\label{appendix-sec-lemptasdpoptimality}

\lemptasdpoptimality*

\begin{proof} 
As $\hat{\sigma}'_i(j)=\hat{\sigma}^*_{i-1}(j)$ for all $j \in J \setminus \left( \; B^{(t)}_i \cup J^{(> i)}\; \right)$, it suffices to compare the assignments $\hat{\sigma}'_i(j)$ and $\hat{\sigma}^*_{i-1}(j)$ make on the job set $B^{(t)}_i \cup J^{(> i)}$ to the machines in $M^{\text{rel}}_i$.

\smallskip
Recall that $\sigma_i$, and also $\hat{\sigma}_i$, are optimal solutions for the instance $\hat{\Pi}_i := (B^{(t)}_i, M, \hat{p}, \hat{c}^{(i)})$ and $\hat{\sigma}'_i(j)=\hat{\sigma}_i(j)$ $\forall j\in B^{(t)}_i$. 
As $\hat{\sigma}_i$ schedules $B^{(t)}_i$ only to the relevant machines in $M^{\text{rel}}_i$, it follows that $\hat{\sigma}_i$, and also $\hat{\sigma}'_i$, are optimal solutions for the instance $(B^{(t)}_i, M^{\text{rel}}_i, \hat{p}, \hat{c}^{(i)})$.
Since $\hat{\sigma}^*_{i-1}$ is a feasible schedule for $(B^{(t)}_i \cup J^{(> i)}, M^{\text{rel}}_i, \hat{p}, \hat{c}^{(i)})$, it follows that $T_{\hat{\Psi}}(\hat{\sigma}'_i) \; \le \; T_{\hat{\Psi}}(\hat{\sigma}^*_{i-1})$.
\end{proof}
%


\subsection{Proof of~\autoref{lem:relevant-machine-lower-bound}}

\label{appendix-sec-lemrelevantmachinelowerbound}

\lemrelevantmachinelowerbound*

\begin{proof}
Suppose that $\hat{\sigma}'_i(j) \neq q$ and $\hat{\sigma}^*(j) \neq q$ for all $j \in \bigcup_{0\le k\le i} B^{(t)}_k$.
Then $q \notin M^{\text{rel}}_0$ and by the design of~\autoref{alg:ptas-proof-assign-small-jobs}, no job in $j \in \bigcup_{0\le k\le i} B^{(t)}_k$ can be assigned to $q$ in $\hat{\sigma}^*_0$.
This implies that $q \notin M^{\text{rel}}_1$.
Repeating the same argument gives that $q \notin M^{\text{rel}}_k$ for all $0\le k\le i$, which is a contradiction to the assumption that $q \in M^{\text{rel}}_i$.
\end{proof}

\end{appendix}

\end{document}